\documentclass[conference]{IEEEtran}
\usepackage{latexsym, graphicx, pstricks, pst-node, url, subfig, graphics, epic, bm}
\usepackage{amssymb}
\usepackage{amsthm}
\usepackage{amsmath}
\usepackage{setspace}
\newtheorem{fact}{Fact}
\newtheorem{definition}{Definition}

\newtheorem{lemma}{Lemma}
\newtheorem{theorem}{Theorem}

\newtheorem{claim}{Claim}

\begin{document}

\title{Computing the Ball Size of Frequency Permutations under Chebyshev Distance}

\author{\IEEEauthorblockN{Min-Zheng Shieh}
\IEEEauthorblockA{Department of Computer Science\\National Chiao Tung University\\
1001 University Road\\
Hsinchu City, Taiwan\\
Email: mzhsieh@csie.nctu.edu.tw}
\and
\IEEEauthorblockN{Shi-Chun Tsai}
\IEEEauthorblockA{Department of Computer Science\\National Chiao Tung University\\
1001 University Road\\
Hsinchu City, Taiwan\\
Email: sctsai@csie.nctu.edu.tw}
}

\date{}

\maketitle

\begin{abstract}
Let $S_n^\lambda$ be the set of all permutations over the multiset $\{\overbrace{1,\dots,1}^{\lambda},\dots,\overbrace{m,\dots,m}^\lambda\}$ where $n=m\lambda$. A frequency permutation array (FPA) of minimum distance $d$ is a subset of $S_n^\lambda$ in which every two elements have distance at least $d$. FPAs have many applications related to error correcting codes. In coding theory, the Gilbert-Varshamov bound and the sphere-packing bound are derived from the size of balls of certain radii. We propose two efficient algorithms that compute the ball size  of frequency permutations under Chebyshev distance. 
Both methods extend previous known results.
The first one runs in $O\left({2d\lambda \choose d\lambda}^{2.376}\log n\right)$ time and $O\left({2d\lambda \choose d\lambda}^{2}\right)$ space.  The second one runs in $O\left({2d\lambda \choose d\lambda}{d\lambda+\lambda\choose \lambda}\frac{n}{\lambda}\right)$ time and $O\left({2d\lambda \choose d\lambda}\right)$ space.
For small constants $\lambda$ and $d$, both are efficient in time and use constant storage space.
\end{abstract}

\section{Introduction}

{\em Frequency permutation arrays} (FPAs) of frequency $\lambda$ and length $n$ are sets of permutations over the multiset $\{\overbrace{1,\dots,1}^{\lambda},\dots,\overbrace{m,\dots,m}^\lambda\}$ where $n=m\lambda$. In particular, FPAs of frequency $\lambda=1$ are called {\em permutation arrays} (PAs). FPA is a special case of Slepian's codes \cite{Slep65} for permutation modulation. 
Many applications of FPAs in various areas, such as power line communication (see \cite{Shum02}, \cite{VH00}, \cite{VHW00} and \cite{Vinck00}), multi-level flash memories (see \cite{Jiang1} and \cite{Jiang2,TS10}) and computer security (see \cite{ST10}), are found recently. 
In many applications, we use FPAs as error correcting codes. It is well-known that the capability against errors of a code is mainly determined by its minimum distance. Similar to traditional codes, the minimum distance of an FPA $C$ is $d$ under metric $\delta$ if $\min_{\rho,\pi\in C; \rho\neq\pi}\delta(\rho,\pi)=d$. A $(\lambda,n,d)$-FPA under some metric $\delta$ is a FPA of frequency $\lambda$ and length $n$ which has minimum distance at least $d$ under $\delta$.  A $(\lambda,n,d)$-FPA is often considered to be better if it has larger cardinality. When we evaluate the quality of certain design of $(\lambda,n,d)$-FPA, we often compare it with the maximum cardinality of $(\lambda,n,d)$-FPAs. Generally speaking, computing the maximum size of codes of certain parameters is hard. In coding theory, the Gilbert-Varshamov bound and the sphere-packing bound are famous lower and upper bounds on the code size, respectively. They are derived from the size of balls of certain radii. We focus on the  efficiency of computing the ball size in this paper.

Shieh and Tsai \cite{ST10} showed that the cardinality of $d$-radius balls can be obtained by 
computing the permanent of a matrix. It is $\#P$-complete to compute the permanent of a general matrix.  
However, the matrix for estimating the ball size of FPAs has a special structure.
Kl{\o}ve \cite{Klove09} used the property and proposed a method to efficiently compute the cardinality of balls of radius $1$ by solving recurrence relations. Meanwhile, based on Schwartz's result \cite{Sch09}, one can compute the size of $d$-radius balls in $O\left({2d \choose d}^{2.376}\log n\right)$ time for $\lambda=1$. 

In this paper, we give two algorithms to compute the cardinality of a $d$-radius ball under Chebyshev distance for $d>1$ and $\lambda >1$. The first one runs in $O\left({2d\lambda \choose d\lambda}^{2.376}\log n\right)$ time and $O\left({2d\lambda \choose d\lambda}^{2}\right)$ space, and the second one in $O\left({2d\lambda \choose d\lambda}{d\lambda+\lambda\choose \lambda}\frac{n}{\lambda}\right)$ time and $O\left({2d\lambda \choose d\lambda}\right)$ space. These algorithms are generalization of Schwartz's result \cite{Sch09}. They are efficient in time and space when $d$ and $\lambda$ are small fixed  constants. 

The rest of the paper is organized as follows. In section 2, we define some notations. In section 3, we introduce a recursive algorithm {\sc EnumV} for enumerating permutations in a $d$-radius ball. Then, based on the property of {\sc EnumV}, we give two methods to obtain the ball size efficiently in section 4. We compare previous results with ours in section 5. Then, we conclude this paper briefly.

\section{Notations}

We set $n=m \lambda$ throughout the work unless stated otherwise. 
For positive integers $a$ and $b$ with $ a< b$, $[a]$ represents the set $\{1,\ldots,a\}$ and
$[a, b]$ represents $\{ a, a+1,\dots, b-1, b\}$. For convenience, $(-\infty,b]$ represents the set of integers at most $b$.
The Chebyshev distance of  two $k$-dimensional vectors $\bm{x}$ and $\bm{y}$ is defined as $d_{\max}(\bm{x},\bm{y})=\max_{i\in [k]} |x_i -y_i|$, where $x_i$ and $y_i$ are the $i$-th entries of $\bm{x}$ and $\bm{y}$ respectively. For permutations $\bm{x}$ and $\bm{y}$, they are said to be $d$-close to each other under metric $\delta(\cdot,\cdot)$ if $\delta(\bm{x},\bm{y})\le d$. 
We use $S_n^{\lambda}$ to denote the set of all frequency permutations with each symbol appearing $\lambda$ times.
The identity frequency permutation $\bm{e}$ in $S_{n}^\lambda$ is $(\overbrace{1,\dots,1}^\lambda,\dots,\overbrace{m,\dots,m}^\lambda)$, i.e.,
 the $i$-th entry of $\bm{e}$ is $e_i=\left\lceil\frac{i}{\lambda}\right\rceil$.
   A {\em partial frequency permutation} can be derived from a frequency permutation  in
$S_{n}^\lambda$ with some entries replaced with $*$.  The symbol  $*$ does not contribute to the
distance.  I.e., the distance between two $k$-dimensional partial frequency permutations,  $\bm{x}$ and $\bm{y}$, is defined as $d_{\max}(\bm{x},\bm{y})=\max_{i\in [k], x_i\neq *, y_i\neq *} |x_i -y_i|$.

Under Chebyshev distance, a ball of radius $r$ centered at $\pi$ is defined as $B(r,\pi)=\{\rho:d_\infty(\rho,\pi)\le r\}$. We can obtain $B(r,\pi)$ from $B(r,\bm{e})$ by rearranging the indices of the entries, therefore $|B(r,\pi)|=|B(r,\bm{e})|$ for $\pi\in S_n^\lambda$. Let $V_\infty(\lambda,n,d)$ be the size of a ball of radius $d$ in $S_n^\lambda$ under Chebyshev distance.


\section{Enumerate $d$-close permutations}

In this section, we give a recursive algorithm to enumerate all frequency permutations in $B(d,\bm{e})$. First, we investigate $\bm{e}$ closely. 
\[\begin{array}{|c||c|c|c|c|c|c|c|c|c|c|c|}\hline
i&\cdots&k\lambda-\lambda&k\lambda-\lambda+1&\cdots&k\lambda&k\lambda+1&\cdots\\\hline
\bm{e}_i&\cdots&k-1&k&\cdots&k&k+1&\cdots\\\hline
\end{array}\]
Observe that  symbol $k$ appears at the $(k\lambda-\lambda+1)$-th, $\dots$, $(k\lambda)$-th positions in $\bm{e}$.  Therefore, $\pi$ is $d$-close to $\bm{e}$ if and only if $\pi_{k\lambda-\lambda+1},\dots,\pi_{k\lambda}\in[k-d,k+d]$. 
(Note that we simply ignore the non-positive  values when $k \le d$.)
In other words,  $d_{\max}(\bm{e},\pi)\le d$ if and only if symbol $k$ only appear in the $(k\lambda-d\lambda-\lambda+1)$-th, $\dots$, $(k\lambda+d\lambda)$-th positions of $\pi$. This observation leads us  to define the shift operator $\oplus$.

\begin{definition} 
For an integer set $S$ and an integer $z$, define $S\oplus z=\{s+z:s\in S\}$. 
\end{definition}

\begin{fact}\label{shift}
Suppose that $\pi$ is $d$-close to $\bm{e}$, then $\pi_i=k$ implies $i\in[-d\lambda+1,d\lambda+\lambda]\oplus(k\lambda-\lambda)$. 
\end{fact}
The above fact is  useful  to capture the frequency permutations that are $d$-close to $\bm{e}$. 
Note that the set $S=[-d\lambda+1,d\lambda+\lambda]$ is independent of $k$.
We give a recursive algorithm {\sc EnumV}$_{\lambda,n,d}$ in figure \ref{enumv}. It enumerates all frequency permutations $\pi\in S_{n}^\lambda$ in $B(d,\bm{e})$. It is a depth-first-search style algorithm. Basically, it first tries to put $1$'s into $\lambda$ proper vacant positions of $\pi$. Then, it tries to put $2$'s, $\dots$, $m$'s into the partial frequency permutations recursively.  According to fact \ref{shift},   symbol $k$ is assigned to positions of indices in $[-d\lambda+1,d\lambda+\lambda]\oplus(k\lambda-\lambda)$, and  these positions are 
said to be {\em valid} for $k$. 
\begin{figure}
\begin{tabbing}
{\sc EnumV}$_{\lambda,n,d}(k,P)$\\
1. ~\=if \=$k\le m$ then\\
2.\>\>for \=each partition $(X,X')$ of $P$ with $|X|=\lambda$ do\\
3.\>\>\>if \=$X'\cap(-\infty,-d\lambda+\lambda]=\emptyset$ then \\
\>\>\>// Make sure it is a proper partition\\
4.\>\>\>\>for \=$i\in X\oplus{(k\lambda-\lambda)}$ do\\
5.\>\>\>\>\>$\pi_i\leftarrow k$; // Assign $k$ to the $i$-th position\\
6.\>\>\>\>$Y\leftarrow (X'\oplus(-\lambda))\cup[d\lambda+1,d\lambda+\lambda]$;\\
7.\>\>\>\>{\sc EnumV}$_{\lambda,n,d}(k+1,Y)$; \\
8.\>\>\>\>for \=$i\in X\oplus(k\lambda-\lambda)$ do\\
9.\>\>\>\>\>$\pi_i\leftarrow0$; // Reset $\pi_i$ to be vacant\\
10.\> else\\
11.\>\> if $P=[d\lambda+\lambda]$ then output $(\pi_1,\dots,\pi_n)$;
\end{tabbing}
\caption{{\sc EnumV}$_{\lambda,n,d}(k,P)$}\label{enumv}
\end{figure}

{\sc EnumV}$_{\lambda,n,d}$ takes an integer $k$ and a subset $P$ of $[-d\lambda+1,d\lambda+\lambda]$ as its input, and {\sc EnumV}$_{\lambda,n,d}$ uses an $(n+2d\lambda)$-dimensional vector $\pi$ as a global variable. For convenience, we extend the index set of $\pi$ to $[-d\lambda+1,n+d\lambda]$ and every entry of $\pi$ is initialized to $0$, which indicates that the entry is vacant. 
We use $P$ to trace the indices of valid vacant positions  for  symbol $k$, and the set of such positions is exactly $P\oplus(k\lambda-\lambda)$. 

We  call  {\sc EnumV}$_{\lambda,n,d}(1,[d\lambda+\lambda])$ to enumerate all frequency permutations which are $d$-close to $\bm{e}$. During the enumeration, {\sc EnumV}$_{\lambda,n,d}(k,P)$ assigns symbol $k$ into some $\lambda$ positions, indexed by $X\oplus(k\lambda-\lambda)$, of a partial frequency permutation, then it recursively invokes {\sc EnumV}$_{\lambda,n,d}(k+1,X'\oplus(-\lambda))\cup[d\lambda+1,d\lambda+\lambda])$, where $X$ and $X'$ form a partition of $P$ and $|X|=\lambda$. After the recursive call is done, {\sc EnumV}$_{\lambda,n,d}(k,P)$ reset positions indexed by $X\oplus(k\lambda-\lambda)$ as vacant. Then, it repeats to search another choice of $\lambda$ positions until all possible combinations of $\lambda$ positions are investigated. For $k=m+1$, {\sc EnumV}$_{\lambda,n,d}(k,P)$ outputs $\pi$ if $P=[d\lambda+\lambda]$.
Given $n=\lambda m$, $k$ is initialized to 1 and $P$ to $[d\lambda+\lambda]$, we have
the following claims.
\begin{claim}\label{max}
In each of the recursive call of {\sc EnumV}$_{\lambda,n,d}$, in line 6 we have $\max(Y)=d\lambda+\lambda$.
\begin{proof}
By induction, it is clear for $k=1$.   Suppose $\max(P)=d\lambda+\lambda$.
Since $Y=(X'\oplus(-\lambda))\cup[d\lambda+1,d\lambda+\lambda]$ and $\max(X')
\le d\lambda+\lambda$, we have  $\max(Y)=d\lambda+\lambda$.
\end{proof}
\end{claim}

\begin{claim}\label{vacant}
In line 6, for each $k\in[m+1]$ and each $i\in Y\oplus((k+1)\lambda-\lambda)$, we have $\pi_i=0$.
\begin{proof}
We prove it by induction on $k$. It is clear for $k=1$.  Assume the claim is true up to $k< m+1$, i.e.,
for each $i\in P\oplus(k\lambda-\lambda)$, $\pi_i =0$. Now, consider the following scenario, {\sc EnumV}$_{\lambda,n,d}(k,P)$ invokes {\sc EnumV}$_{\lambda,n,d}(k+1,Y)$. 

Since  $Y=(X'\oplus(-\lambda))\cup[d\lambda+1,d\lambda+\lambda]$, we have
$Y\oplus((k+1)\lambda-\lambda)=(X'\oplus(k\lambda-\lambda)) \cup[k\lambda+d\lambda+1,k\lambda+d\lambda+\lambda]$. 
While $X'\subset P$ and $[k\lambda+d\lambda+1,k\lambda+d\lambda+\lambda]$ are new vacant positions, it is clear
$\pi_i=0$ in these entries. 
\end{proof}
\end{claim}

\begin{claim} \label{range}
In each recursive call of {\sc EnumV}$_{\lambda,n,d}(k,P)$, $P$ must be a subset of $[-d\lambda+1,d\lambda+\lambda]$ of cardinality $d\lambda+\lambda$. It implies, $|e_i-k|\le d$ for $i\in P\oplus(k\lambda -\lambda)$.

\end{claim}
\begin{proof}
We prove it by induction on $k$. For $k=1$, $P$ is $[d\lambda+\lambda]$, and the claim is obvious.  Assume the claim is true up to $k$, and {\sc EnumV}$_{\lambda,n,d}(k,P)$ invokes {\sc EnumV}$_{\lambda,n,d}(k+1,Y)$. Thus $Y=(X'\oplus(-\lambda))\cup[d\lambda+1,d\lambda+\lambda]$. Due to the constraint on $X'$ in line 3 and the induction hypothesis, we have $X'\oplus(-\lambda)\subseteq[-d\lambda+1,d\lambda]$. We conclude that $Y\subseteq [-d\lambda+1,d\lambda]\cup[d\lambda+1,d\lambda+\lambda]$ and $|Y|=|X'|+\lambda=d\lambda+\lambda$.  Since
$[-d\lambda+1,d\lambda+\lambda]\oplus(k\lambda)=[k\lambda -d\lambda+1,k\lambda +d\lambda+\lambda]$, we know $\bm{e}$ has values from $[k-d+1, k+d+1]$ in these positions. I.e., $|e_i-(k+1)|\le d$ for $i\in Y\oplus(k\lambda)$.
Hence, the claim is true.
\end{proof}

\begin{claim}\label{noneg}
At the beginning of the  invocation of {\sc EnumV}$_{\lambda,n,d}(k,P)$, $i\in P\oplus(k\lambda-\lambda)$ implies $i>0$.
\end{claim}
\begin{proof}
It is clear for $k=1$. Observe that  $\min(Y)\ge\min(P)-\lambda$.  Since,
$(\min(P)-\lambda)\oplus(k\lambda)=\min(P)\oplus (k\lambda-\lambda)$, 
the claim holds for $k>1$.
\end{proof}

\begin{claim}\label{progress}
For $k\in[m]$, when {\sc EnumV}$_{\lambda,n,d}(k,P)$ invoke {\sc EnumV}$_{\lambda,n,d}(k+1,Y)$ in line 7, there are exactly $\lambda$ entries of $\pi$ equal $i$ for $i\in[k-1]$.
\end{claim}
\begin{proof}
It is implied by lines 4 and 5. 
\end{proof}

\begin{lemma}\label{trace}
At the beginning of the execution of {\sc EnumV}$_{\lambda,n,d}(k,P)$, $P\oplus(k\lambda-\lambda)=\{i:i>0\wedge \pi_i=0\wedge i\in[-d\lambda+1,d\lambda+\lambda]\oplus(k\lambda-\lambda)\}$.
\end{lemma}
\begin{proof}
The lemma holds by claims \ref{vacant},  \ref{range}, and  \ref{noneg}.
\end{proof}

Let $\pi$ be the output of {\sc EnumV}$_{\lambda,n,d}(1,[d\lambda+\lambda])$. These facts ensure that  $\pi\in S_{n}^\lambda$ and $d_{\max}(\bm{e},\pi)\le d$. 

\begin{lemma}\label{enumerate}
For $k\in[m+1]$, let $\tau_k$ be a partial frequency permutation $d$-close to $\bm{e}$ and with each symbol $1,\dots, k-1$ appearing exactly $\lambda$ times in $\tau_k$.
For $i=1,\dots,n$, let $\pi_i=(\tau_k)_i$, if  $(\tau_k)_i \neq *$.  Let $\pi_i=0$ for the other $i$.
Every frequency permutation $\rho$ in $S_n^\lambda$, generated by
 {\sc EnumV}$_{\lambda,n,d}(k,P)$,   satisfies the following conditions:
\begin{enumerate}
\item  $\rho$ is consistent with $\tau_k$ over the entries with symbols $1,\dots, k-1$.
\item $\rho$ is $d$-close to $\bm{e}$.
\end{enumerate}
\end{lemma}
\begin{proof}
We prove it by reverse induction. First, we consider the case $k=m+1$.  By claim \ref{progress}, every symbol appears exactly $\lambda$ times in $\tau_{m+1}$ and $\pi$. By claim \ref{vacant} and claim \ref{noneg}, we know that all of $\pi_1,\dots,\pi_n$ are nonzero if and only if $P=[d\lambda+\lambda]$. There are two possible cases:
\begin{itemize}
\item $P=[d\lambda+\lambda]$: By claim \ref{range},  $(\pi_1,\dots,\pi_n)$ is the only frequency permutation in $S_n^\lambda$  satisfying both conditions. 
\item $P\neq[d\lambda+\lambda]$: Then there is some $i\in [n]$ such that $\pi_i=0$. But there is  no frequency permutation $\rho$ in $S_n^\lambda$ with $\rho_i=m+1>m$.  This means $\pi$ is not well
assigned.
\end{itemize}
Note that {\sc EnumV}$_{\lambda,n,d}(k,P)$ outputs $\pi$ only if $P=[d\lambda+\lambda]$, otherwise there is no output. Hence, the claim is true for $k=m+1$. Assume the claim is true 
down to  $k+1$. For $k$, by lemma \ref{trace}, $P\oplus(k\lambda-\lambda)$ is exactly the set of all  positions which are vacant and valid for  $k$. In order to enumerate permutations satisfying the second condition, we must assign $k$'s into the valid positions. Therefore, we just need to try all possible choices of $\lambda$-element subset $X\oplus(k\lambda-\lambda)\subseteq P\oplus(k\lambda-\lambda)$. 
Line 3 of  {\sc EnumV}$_{\lambda,n,d}(k,P)$ ensures that $X$ is properly selected.
Then symbol $k$ is assigned to $X\oplus(k\lambda-\lambda)$ and we have a new
partial frequency permutation $\tau_{k+1}$.  By induction hypothesis, it is clear that the generated frequency permutations satisfy
both conditions.
\end{proof}

We have the following theorem as an immediate result of lemma \ref{enumerate}.

\begin{theorem}\label{EnumTHM}
{\sc EnumV}$_{\lambda,n,d}(1,[d\lambda+\lambda])$ enumerates exactly all $d$-close frequency permutations to $\bm{e}$ in $S_{n}^\lambda$.
\end{theorem}
\begin{proof}
Let $\rho$ be a $d$-close frequency permutation to $\bm{e}$ in $S_{n}^\lambda$.
There is a corresponding sequence of partial frequency permutations $\tau_1,\dots, \tau_{m+1}$ $d$-close to
$\bm{e}$.  Note that $\tau_1$ has $*$ in all of its entries and $\tau_{m+1}=\rho$.  
By lemma  \ref{enumerate}, $\rho$ will be enumerated eventually.
\end{proof}

\section{Computing $V_\infty(\lambda,n,d)$}

The number of elements generated by {\sc EnumV}$_{\lambda,n,d}(1,[d\lambda+\lambda])$ is clearly $V_\infty(\lambda,n,d)$. However, the enumeration is not efficient, since $V_\infty(\lambda,n,d)$ is usually a very large number. In this section, we give two efficient  implementations to compute $V_\infty(\lambda,n,d)$.
Especially,
$V_\infty(\lambda,n,d)$ can be computed in polynomial time for constant $d$ and $\lambda$.

From the algorithm {\sc EnumV}, we see that whether {\sc EnumV}$_{\lambda,n,d}(k,P)$ invokes {\sc EnumV}$_{\lambda,n,d}(k+1,Y)$ or not depends only on $k$, $P$ and $Y$. During the execution of {\sc EnumV}$_{\lambda,n,d}(1,[d\lambda+\lambda])$,   {\sc EnumV}$_{\lambda,n,d}(k,P)$ is invoked recursively only when $[d\lambda+1,d\lambda+\lambda] \subset P \subset [-d\lambda+1,d\lambda+\lambda]$, due to line 6. Therefore, we can construct a directed acyclic graph $G_{\lambda,n,d}=\langle V_G,E_G\rangle$ where 
\begin{itemize}
\item $V_G=\{(k,U):k\in[m+1],U\subset[-d\lambda+1,d\lambda]$ and $|U|=d\lambda\}$.
\item $((k,U),(k+1,V))\in E_G$ if and only if {\sc EnumV}$_{\lambda,n,d}(k,U\cup [d\lambda+1,d\lambda+\lambda])$ invokes {\sc EnumV}$_{\lambda,n,d}(k+1,V\cup[d\lambda+1,d\lambda+\lambda])$. 
\end{itemize} 
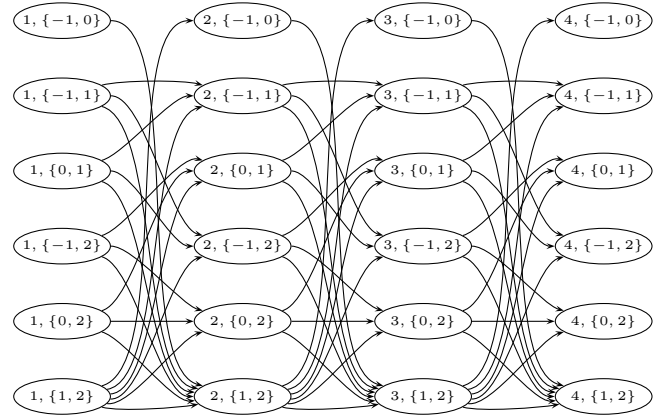
\begin{figure}[h]
\center
\begin{pspicture}(-4.2,-0.2)(4.2,5.2)
\tiny
\psset{linewidth=0.4pt}
\rput(-3.6,0){\ovalnode{A1}{\makebox[0.75cm]{$1,\{1,2\}$}}}
\rput(-3.6,1){\ovalnode{B1}{\makebox[0.75cm]{$1,\{0,2\}$}}}
\rput(-3.6,2){\ovalnode{C1}{\makebox[0.75cm]{$1,\{-1,2\}$}}}
\rput(-3.6,3){\ovalnode{D1}{\makebox[0.75cm]{$1,\{0,1\}$}}}
\rput(-3.6,4){\ovalnode{E1}{\makebox[0.75cm]{$1,\{-1,1\}$}}}
\rput(-3.6,5){\ovalnode{F1}{\makebox[0.75cm]{$1,\{-1,0\}$}}}
\rput(-1.2,0){\ovalnode{A2}{\makebox[0.75cm]{$2,\{1,2\}$}}}
\rput(-1.2,1){\ovalnode{B2}{\makebox[0.75cm]{$2,\{0,2\}$}}}
\rput(-1.2,2){\ovalnode{C2}{\makebox[0.75cm]{$2,\{-1,2\}$}}}
\rput(-1.2,3){\ovalnode{D2}{\makebox[0.75cm]{$2,\{0,1\}$}}}
\rput(-1.2,4){\ovalnode{E2}{\makebox[0.75cm]{$2,\{-1,1\}$}}}
\rput(-1.2,5){\ovalnode{F2}{\makebox[0.75cm]{$2,\{-1,0\}$}}}
\rput(1.2,0){\ovalnode{A3}{\makebox[0.75cm]{$3,\{1,2\}$}}}
\rput(1.2,1){\ovalnode{B3}{\makebox[0.75cm]{$3,\{0,2\}$}}}
\rput(1.2,2){\ovalnode{C3}{\makebox[0.75cm]{$3,\{-1,2\}$}}}
\rput(1.2,3){\ovalnode{D3}{\makebox[0.75cm]{$3,\{0,1\}$}}}
\rput(1.2,4){\ovalnode{E3}{\makebox[0.75cm]{$3,\{-1,1\}$}}}
\rput(1.2,5){\ovalnode{F3}{\makebox[0.75cm]{$3,\{-1,0\}$}}}
\rput(3.6,0){\ovalnode{A4}{\makebox[0.75cm]{$4,\{1,2\}$}}}
\rput(3.6,1){\ovalnode{B4}{\makebox[0.75cm]{$4,\{0,2\}$}}}
\rput(3.6,2){\ovalnode{C4}{\makebox[0.75cm]{$4,\{-1,2\}$}}}
\rput(3.6,3){\ovalnode{D4}{\makebox[0.75cm]{$4,\{0,1\}$}}}
\rput(3.6,4){\ovalnode{E4}{\makebox[0.75cm]{$4,\{-1,1\}$}}}
\rput(3.6,5){\ovalnode{F4}{\makebox[0.75cm]{$4,\{-1,0\}$}}}
\psset{nodesep=0pt} 
\nccurve[angleA=-15,angleB=190,ncurv=0.4]{->}{A1}{A2} 
\nccurve[angleA=-9,angleB=195,ncurv=0.4]{->}{A1}{B2}
\nccurve[angleA=-3,angleB=195,ncurv=0.4]{->}{A1}{C2}
\nccurve[angleA=3,angleB=195,ncurv=0.4]{->}{A1}{D2} 
\nccurve[angleA=9,angleB=195,ncurv=0.4]{->}{A1}{E2}
\nccurve[angleA=15,angleB=180,ncurv=0.4]{->}{A1}{F2} 
\nccurve[angleA=-15,angleB=185,ncurv=0.4]{->}{B1}{A2} 
\nccurve[angleA=0,angleB=180,ncurv=0.4]{->}{B1}{B2} 
\nccurve[angleA=15,angleB=180,ncurv=0.4]{->}{B1}{D2} 
\nccurve[angleA=-15,angleB=180,ncurv=0.4]{->}{C1}{A2} 
\nccurve[angleA=0,angleB=165,ncurv=0.4]{->}{C1}{B2} 
\nccurve[angleA=15,angleB=165,ncurv=0.4]{->}{C1}{D2} 
\nccurve[angleA=-15,angleB=175,ncurv=0.4]{->}{D1}{A2} 
\nccurve[angleA=0,angleB=180,ncurv=0.4]{->}{D1}{C2} 
\nccurve[angleA=15,angleB=180,ncurv=0.4]{->}{D1}{E2} 
\nccurve[angleA=-15,angleB=170,ncurv=0.4]{->}{E1}{A2} 
\nccurve[angleA=0,angleB=165,ncurv=0.4]{->}{E1}{C2} 
\nccurve[angleA=15,angleB=165,ncurv=0.4]{->}{E1}{E2} 
\nccurve[angleA=0,angleB=165,ncurv=0.4]{->}{F1}{A2} 

\nccurve[angleA=-15,angleB=190,ncurv=0.4]{->}{A2}{A3} 
\nccurve[angleA=-9,angleB=195,ncurv=0.4]{->}{A2}{B3}
\nccurve[angleA=-3,angleB=195,ncurv=0.4]{->}{A2}{C3}
\nccurve[angleA=3,angleB=195,ncurv=0.4]{->}{A2}{D3} 
\nccurve[angleA=9,angleB=195,ncurv=0.4]{->}{A2}{E3}
\nccurve[angleA=15,angleB=180,ncurv=0.4]{->}{A2}{F3} 
\nccurve[angleA=-15,angleB=185,ncurv=0.4]{->}{B2}{A3} 
\nccurve[angleA=0,angleB=180,ncurv=0.4]{->}{B2}{B3} 
\nccurve[angleA=15,angleB=180,ncurv=0.4]{->}{B2}{D3} 
\nccurve[angleA=-15,angleB=180,ncurv=0.4]{->}{C2}{A3} 
\nccurve[angleA=0,angleB=165,ncurv=0.4]{->}{C2}{B3} 
\nccurve[angleA=15,angleB=165,ncurv=0.4]{->}{C2}{D3} 
\nccurve[angleA=-15,angleB=175,ncurv=0.4]{->}{D2}{A3} 
\nccurve[angleA=0,angleB=180,ncurv=0.4]{->}{D2}{C3} 
\nccurve[angleA=15,angleB=180,ncurv=0.4]{->}{D2}{E3} 
\nccurve[angleA=-15,angleB=170,ncurv=0.4]{->}{E2}{A3} 
\nccurve[angleA=0,angleB=165,ncurv=0.4]{->}{E2}{C3} 
\nccurve[angleA=15,angleB=165,ncurv=0.4]{->}{E2}{E3} 
\nccurve[angleA=0,angleB=165,ncurv=0.4]{->}{F2}{A3} 

\nccurve[angleA=-15,angleB=190,ncurv=0.4]{->}{A3}{A4} 
\nccurve[angleA=-9,angleB=195,ncurv=0.4]{->}{A3}{B4}
\nccurve[angleA=-3,angleB=195,ncurv=0.4]{->}{A3}{C4}
\nccurve[angleA=3,angleB=195,ncurv=0.4]{->}{A3}{D4} 
\nccurve[angleA=9,angleB=195,ncurv=0.4]{->}{A3}{E4}
\nccurve[angleA=15,angleB=180,ncurv=0.4]{->}{A3}{F4} 
\nccurve[angleA=-15,angleB=185,ncurv=0.4]{->}{B3}{A4} 
\nccurve[angleA=0,angleB=180,ncurv=0.4]{->}{B3}{B4} 
\nccurve[angleA=15,angleB=180,ncurv=0.4]{->}{B3}{D4} 
\nccurve[angleA=-15,angleB=180,ncurv=0.4]{->}{C3}{A4} 
\nccurve[angleA=0,angleB=165,ncurv=0.4]{->}{C3}{B4} 
\nccurve[angleA=15,angleB=165,ncurv=0.4]{->}{C3}{D4} 
\nccurve[angleA=-15,angleB=175,ncurv=0.4]{->}{D3}{A4} 
\nccurve[angleA=0,angleB=180,ncurv=0.4]{->}{D3}{C4} 
\nccurve[angleA=15,angleB=180,ncurv=0.4]{->}{D3}{E4} 
\nccurve[angleA=-15,angleB=170,ncurv=0.4]{->}{E3}{A4} 
\nccurve[angleA=0,angleB=165,ncurv=0.4]{->}{E3}{C4} 
\nccurve[angleA=15,angleB=165,ncurv=0.4]{->}{E3}{E4} 
\nccurve[angleA=0,angleB=165,ncurv=0.4]{->}{F3}{A4} 
\end{pspicture}
\caption{Graph $G_{2,6,1}$}\label{G261}
\end{figure}
For example, figure \ref{G261} shows the structure of $G_{2,6,1}$.

$V_\infty(\lambda,n,d)$ equals the number of invocations of {\sc EnumV}$_{\lambda,n,d}(m+1,[d\lambda+\lambda])$. With this observation, $V_\infty(\lambda,n,d)$ also equals the number of paths from $(1,[d\lambda])$ to $(m+1,[d\lambda])$ in $G_{\lambda,n,d}$. By the definition of $G_{\lambda,n,d}$, it is a directed acyclic graph. The number of paths from one vertex to another in a directed acyclic graph can be computed in $O(|V|+|E|)$, where $|V|=(m+1){2d\lambda\choose d\lambda}$ and $|E|=O(|V|^2)$.  So  $V_\infty(\lambda,n,d)$ can be calculated in polynomial time with respect to $n$ if $\lambda$ and $d$ are constants.

The computation actually can be done in $O(\log n)$ for constant $\lambda$ and $d$. Define $H_{\lambda,d}=\langle V_H, E_H\rangle$ where $V_H=\{P:|P|=d\lambda\mbox{ and }P\subseteq [-d\lambda+1,d\lambda]\}$ and $(P,P')\in E_H$ if and only if there is some $k\in[m]$ such that {\sc EnumV}$_{\lambda,n,d}(k,P\cup[d\lambda+1,d\lambda+\lambda])$ invokes {\sc EnumV}$_{\lambda,n,d}(k+1,P'\cup [d\lambda+1,d\lambda+\lambda])$.
Note that $|V_H|={2d\lambda\choose d\lambda}$.
Figure \ref{H21} shows $H_{2,1}$ as an example.
\begin{figure}[h]
\center
\begin{pspicture}(-4,-2)(4,2)
\rput(-3.5,0){\ovalnode{A}{\makebox[0.8cm]{$\{1,2\}$}}}
\rput(-1.5,1.5){\ovalnode{B}{\makebox[0.8cm]{$\{0,2\}$}}}
\rput(-1.5,-1.5){\ovalnode{C}{\makebox[0.8cm]{$\{-1,2\}$}}}
\rput(1.5,1.5){\ovalnode{D}{\makebox[0.8cm]{$\{0,1\}$}}}
\rput(1.5,-1.5){\ovalnode{E}{\makebox[0.8cm]{$\{-1,1\}$}}}
\rput(3.5,0){\ovalnode{F}{\makebox[0.8cm]{$\{-1,0\}$}}}
\psset{nodesep=0pt} 
\nccurve[angleA=135,angleB=225,ncurv=3]{->}{A}{A} \ncarc{->}{A}{B} \ncarc{->}{A}{C} \ncarc{->}{A}{D} \ncarc{->}{A}{E} \ncarc{->}{A}{F}
\ncarc[arcangle=-16]{->}{B}{A}
\nccurve[angleA=135,angleB=180,ncurv=3]{->}{B}{B}
\ncarc{->}{B}{D}
\ncarc{->}{C}{A}
\ncarc{->}{C}{B}
\ncarc{->}{C}{D}
\ncarc[arcangle=0]{->}{D}{A}
\ncarc{->}{D}{C}
\ncarc{->}{D}{E}
\ncarc[arcangle=0]{->}{E}{A}
\ncarc{->}{E}{C}
\nccurve[angleA=45,angleB=360,ncurv=3]{->}{E}{E}
\ncarc[arcangle=0]{->}{F}{A}
\end{pspicture}
\caption{Graph $H_{2,1}$}\label{H21}
\end{figure}
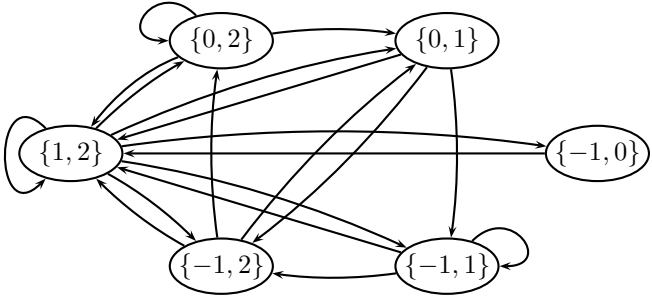

\begin{theorem}\label{logtime}
$V_\infty(\lambda,n,d)$ can be computed in $O(\log n)$ time for constant $d$ and $\lambda$.
\end{theorem}

\begin{proof}
 Observe that the value of $k\in[m]$ is independent of  the invocation of {\sc EnumV}$_{\lambda,n,d}(k+1,P'\cup[d\lambda+1,d\lambda+\lambda])$ by {\sc EnumV}$_{\lambda,n,d}(k,P\cup[d\lambda+1,d\lambda+\lambda])$ , where $P$ and $P'\subset[-d\lambda+1,d\lambda]$ with $|P|=|P'|=d\lambda$. Therefore, the number of paths of length $m$ from $[d\lambda]$ to itself in $H_{\lambda,d}$ is equal to the number of paths from $(1,[d\lambda])$ to $(m+1,[d\lambda])$ in $G_{\lambda,n,d}$. 
 
Let $V_H=\left\{v_1,\dots,v_{|V_H|}\right\}$, where $v_1=[1,d\lambda]$. 
The number of paths of length $m$ from $v_1$ to $v_1$ is the first entry of the first column of the $m$-th power of $A_H$, where $A_H$ is the adjacency matrix of $H_{\lambda,d}$. Since $m$-th power can be computed in $O\left(f\left({2d\lambda\choose d\lambda}\right)\log m\right)$, where $O(f(x))$ is the time cost of multiplying two $x\times x$ matrices. It is well-know that  $f(x)=O(x^{2.376})$ by the Coppersmith-Winograd algorithm.  With constants $\lambda$ and $d$, $V_{\lambda,n,d}$ can be found in $O(\log n)$ time.
\end{proof}

However, the space to store the adjacency matrix $A_H$, $\Omega\left({2d\lambda \choose d\lambda}^{2}\right)$, can be too large to execute the $O(\log n)$-time algorithm.  For example, by setting $d=3$ and $\lambda=3$, we need at least ${18 \choose 9}^{2}\approx 2.36\times 10^9$ entries to store $A_H^{m}$. This makes the constant factor extraordinarily large in the proof of theorem \ref{logtime}. Hence, we provide an alternative implementation which runs in $O\left({2d\lambda\choose d\lambda}\cdot{d\lambda+\lambda\choose \lambda}\cdot m\right)$ time and $O({2d\lambda\choose d\lambda})$ space. This allows us to compute more efficiently for the cases with smaller $m$ and larger $d$ and $\lambda$. For example, ${2d\lambda\choose d\lambda}\cdot{d\lambda+\lambda\choose \lambda}\cdot m\approx 10^9$ for $d=3,\lambda=3,$ and $m=100$. To achieve the $O({2d\lambda\choose d\lambda})$-space complexity, we do not store the adjacency matrix $A_H$ in the memory. Since $A_H$ is the adjacency matrix of $H_{\lambda,d}$, for $\bm{y}=(y_1,\dots,y_{|V_H|})$ and $\bm{y}'=A_H\bm{y}=(y'_1,\dots,y'_{|V_H|})$, we have $y'_i=\sum_{(i,j)\in E_H}y_j$. Hence, if enumerating all edges in $E_H$ takes $S$ space and $T$ time, then we can  compute $A_H\bm{y}$ in $O(|V_H|+S)$ space and $O(T)$ time for any $|V_H|$-dimension $\bm{y}$.

\begin{lemma}\label{edge}
$|E_H|\le|V_H|\cdot{d\lambda+\lambda\choose \lambda}$ and $E_H$ can be enumerated in $O(d\lambda)$ space and $O(|E_H|)$ time.
\end{lemma}
\begin{proof}
For $P\in V_H$ such that $|P\cap (-\infty, -d\lambda-\lambda]|=r$, $P$ has $d\lambda+\lambda-r\choose \lambda-r$ out-going edges, since every partition $(X,X')$ of $P\cup[d\lambda+1,d\lambda+\lambda]$ satisfies the condition in line 3 if and only if $P\cap (-\infty, -d\lambda-\lambda]\subset X$, i.e., every choice of $(\lambda-r)$-element subset of $P\backslash (-\infty, -d\lambda-\lambda]$ will invoke a recursive call. Since ${d\lambda+\lambda-r\choose \lambda-r}\le{d\lambda+\lambda \choose \lambda}$, the number of edges has an upper bound $|V_H|{d\lambda+\lambda\choose \lambda}$. 

To enumerate all $\lambda$-element subsets of a $(d\lambda+\lambda)$-element set, we need $O(d\lambda+\lambda)=O(d\lambda)$ space and $O({d\lambda+\lambda\choose \lambda})$ time. Since we can recycle the space, the enumeration of edges in $E_H$ can be done in $O(d\lambda)$ space and $O(|E_H|)$ time.
\end{proof}
Now, we give the alternative implementation. 

\begin{theorem}\label{lineartime}
$V_\infty(\lambda,n,d)$ can be computed in $O\left({2d\lambda\choose d\lambda}\cdot{d\lambda+\lambda\choose \lambda}\cdot m\right)$ time and $O({2d\lambda\choose d\lambda})$ space.
\end{theorem}

\begin{proof}
Let $\bm{x}=(1,0,\dots,0)^T$. Since $A_H^m\bm{x}$ is the first column of $A_H^m$, $V_\infty(\lambda,n,d)$ is the first entry of $A_H^m\bm{x}$. The alternative evaluates $A_H^1\bm(x),\dots,A_H^m\bm{x}$ iteratively. Instead of storing the whole adjacency matrix, it only uses two $|V_H|$-dimension vectors $\bm{y}$ and $\bm{y}'$ for storing $A_H^{i}\bm{x}$ and the intermediate result of $A_H^{i+1}\bm{x}$, respectively. Initially, $\bm{y}=\bm{x}$ and $i=0$. We compute $A_H\bm{y}$ by the algorithm described in lemma \ref{edge} and using $\bm{y}'$ to store the intermediate result. Then, we copy the result of $A_H\bm{y}$ back to $\bm{y}$. After $m$ repetitions, we have $\bm{y}=A_H^m\bm{x}$, and the first entry of $\bm{y}$ is $V_\infty(\lambda,n,d)$.
Therefore, the space complexity can be reduced to $O(|V_H|+d\lambda)=O({2d\lambda\choose d\lambda})$. The running time is $m\cdot O(|E_H|)=O\left({2d\lambda\choose d\lambda}\cdot{d\lambda+\lambda\choose \lambda}\cdot m\right)$.
\end{proof}

\section{Comparison with previous results}

In this section, we compare our results with previous ones. Shieh and Tsai \cite{ST10} showed that $(\lambda!)^mV_\infty(\lambda,n,d)$ equals the permanent of 0-1 matrix $A^{(\lambda,n,d)}=(a_{i,j}^{(\lambda,n,d)})$ where $a_{i,j}^{(\lambda,n,d)}=1$ if and only if $\left|\left\lceil\frac{i}{\lambda}\right\rceil-\left\lceil\frac{j}{\lambda}\right\rceil\right|\le d$. For example, 
\[A^{(2,8,1)}=\left(\begin{array}{cccccccc}1&1&1&1&0&0&0&0\\1&1&1&1&0&0&0&0\\1&1&1&1&1&1&0&0\\ 1&1&1&1&1&1&0&0\\0&0&1&1&1&1&1&1\\0&0&1&1&1&1&1&1\\0&0&0&0&1&1&1&1\\0&0&0&0&1&1&1&1\end{array}\right)\]
A naive approach to evaluate the permanent of an $n$-by-$n$ matrix takes $O(n!)$ time. In practice, $\Theta((\lambda!)^mV_\infty(\lambda,n,d))$ time is still required when using a backtracking algorithm. It is clear that both of our method are much more faster.

Kl{\o}ve \cite{Klove09} solved the recurrence of $V_\infty(\lambda,n,d)$, and he gave the value of $V_\infty(\lambda,\lambda m,d)$ for $\lambda\in[10]$, $m\in[20]$, and $d=1$. Schwartz \cite{Sch09} gave an algorithm which can be applied for computing $V_\infty(1,n,d)$. In this paper, we provide solutions to computing $V_\infty(\lambda,n,d)$ for $\lambda>1$ and $d>1$, which is not contained in their works. We list the values of $V_\infty(\lambda,\lambda m,d)$ for $\lambda>1$, $m\in[20]$, $d>1$, and $d\lambda\le 10$.

{\tiny
\begin{spacing}{0.5}
$\\
V_\infty(2,2,2)=1\\
V_\infty(2,4,2)=6\\
V_\infty(2,6,2)=90\\
V_\infty(2,8,2)=786\\
V_\infty(2,10,2)=6139\\
V_\infty(2,12,2)=54073\\
V_\infty(2,14,2)=477228\\
V_\infty(2,16,2)=4113864\\
V_\infty(2,18,2)=35579076\\
V_\infty(2,20,2)=308945881\\
V_\infty(2,22,2)=2679325561\\
V_\infty(2,24,2)=23222971098\\
V_\infty(2,26,2)=201351085146\\
V_\infty(2,28,2)=1745886520422\\
V_\infty(2,30,2)=15137227297027\\
V_\infty(2,32,2)=131243141767393\\
V_\infty(2,34,2)=1137923361184848\\
V_\infty(2,36,2)=9866167034815440\\
V_\infty(2,38,2)=85542686564024352\\
V_\infty(2,40,2)=741681846818742097\\
\\
V_\infty(2,2,3)=1\\
V_\infty(2,4,3)=6\\
V_\infty(2,6,3)=90\\
V_\infty(2,8,3)=2520\\
V_\infty(2,10,3)=45450\\
V_\infty(2,12,3)=669666\\
V_\infty(2,14,3)=9747523\\
V_\infty(2,16,3)=154700569\\
V_\infty(2,18,3)=2502207156\\
V_\infty(2,20,3)=40043708244\\
V_\infty(2,22,3)=632349938520\\
V_\infty(2,24,3)=9986116318524\\
V_\infty(2,26,3)=158192179607364\\
V_\infty(2,28,3)=2509767675626581\\
V_\infty(2,30,3)=39796612230719845\\
V_\infty(2,32,3)=630688880128338378\\
V_\infty(2,34,3)=9994168619297530758\\
V_\infty(2,36,3)=158396161513685960664\\
V_\infty(2,38,3)=2510580301930785916566\\
V_\infty(2,40,3)=39792149406721332018414\\
\\
V_\infty(2,2,4)=1\\
V_\infty(2,4,4)=6\\
V_\infty(2,6,4)=90\\
V_\infty(2,8,4)=2520\\
V_\infty(2,10,4)=113400\\
V_\infty(2,12,4)=3540600\\
V_\infty(2,14,4)=88610850\\
V_\infty(2,16,4)=2044242426\\
V_\infty(2,18,4)=47806940971\\
V_\infty(2,20,4)=1196081134201\\
V_\infty(2,22,4)=30647443460124\\
V_\infty(2,24,4)=784921116539484\\
V_\infty(2,26,4)=19899840884886720\\
V_\infty(2,28,4)=500019936693729120\\
V_\infty(2,30,4)=12551808236761063440\\
V_\infty(2,32,4)=315694279415609776404\\
V_\infty(2,34,4)=7955400980632212027852\\
V_\infty(2,36,4)=200622722060793477132937\\
V_\infty(2,38,4)=5057787000067792980984649\\
V_\infty(2,40,4)=127452627155747602225756890\\
\\
V_\infty(2,2,5)=1\\
V_\infty(2,4,5)=6\\
V_\infty(2,6,5)=90\\
V_\infty(2,8,5)=2520\\
V_\infty(2,10,5)=113400\\
V_\infty(2,12,5)=7484400\\
V_\infty(2,14,5)=361859400\\
V_\infty(2,16,5)=14091630840\\
V_\infty(2,18,5)=489147860970\\
V_\infty(2,20,5)=16420511188146\\
V_\infty(2,22,5)=563209318269379\\
V_\infty(2,24,5)=20416518083009593\\
V_\infty(2,26,5)=758713036253909844\\
V_\infty(2,28,5)=28351365170599079604\\
V_\infty(2,30,5)=1054143198114097909680\\
V_\infty(2,32,5)=38864351069181445164480\\
V_\infty(2,34,5)=1423417411123883479886400\\
V_\infty(2,36,5)=52064892889568503574209920\\
V_\infty(2,38,5)=1906534315066176639758670480\\
V_\infty(2,40,5)=69931615009402042606373019804\\
\\
V_\infty(3,3,2)=1\\
V_\infty(3,6,2)=20\\
V_\infty(3,9,2)=1680\\
V_\infty(3,12,2)=61340\\
V_\infty(3,15,2)=1886431\\
V_\infty(3,18,2)=69496201\\
V_\infty(3,21,2)=2568223000\\
V_\infty(3,24,2)=91712960320\\
V_\infty(3,27,2)=3290467596440\\
V_\infty(3,30,2)=118724053748417\\
V_\infty(3,33,2)=4276273204804217\\
V_\infty(3,36,2)=153904262366842444\\
V_\infty(3,39,2)=5541519231941145440\\
V_\infty(3,42,2)=199545071017172522244\\
V_\infty(3,45,2)=7184755645113714298863\\
V_\infty(3,48,2)=258691998154725997048673\\
V_\infty(3,51,2)=9314545233907934721851472\\
V_\infty(3,54,2)=335381528796576643131475840\\
V_\infty(3,57,2)=12075785123501322139824319056\\
V_\infty(3,60,2)=434802491356562053648077727185\\
\\
V_\infty(3,3,3)=1\\
V_\infty(3,6,3)=20\\
V_\infty(3,9,3)=1680\\
V_\infty(3,12,3)=369600\\
V_\infty(3,15,3)=41480880\\
V_\infty(3,18,3)=3422150780\\
V_\infty(3,21,3)=276888204387\\
V_\infty(3,24,3)=25512718688405\\
V_\infty(3,27,3)=2418264595619240\\
V_\infty(3,30,3)=225661997838758560\\
V_\infty(3,33,3)=20649533952628896000\\
V_\infty(3,36,3)=1889648253594082624960\\
V_\infty(3,39,3)=173699198403114756474600\\
V_\infty(3,42,3)=16001577154624484682748453\\
V_\infty(3,45,3)=1472965856766989578006355117\\
V_\infty(3,48,3)=135481185586476496195656612044\\
V_\infty(3,51,3)=12459839493182349378716705969200\\
V_\infty(3,54,3)=1146141579672729885487800599057600\\
V_\infty(3,57,3)=105440511941055519854115528116882480\\
V_\infty(3,60,3)=9699923367172090411762252385134967844\\
\\
V_\infty(4,4,2)=1\\
V_\infty(4,8,2)=70\\
V_\infty(4,12,2)=34650\\
V_\infty(4,16,2)=5562130\\
V_\infty(4,20,2)=708212251\\
V_\infty(4,24,2)=114774147001\\
V_\infty(4,28,2)=18679465660540\\
V_\infty(4,32,2)=2906167849870600\\
V_\infty(4,36,2)=454904037056013460\\
V_\infty(4,40,2)=71729455730285511001\\
V_\infty(4,44,2)=11285129375761977675001\\
V_\infty(4,48,2)=1773699532985462649188410\\
V_\infty(4,52,2)=278931562239767189408085850\\
V_\infty(4,56,2)=43869015908453746845566145990\\
V_\infty(4,60,2)=6898693708786029238293860809251\\
V_\infty(4,64,2)=1084865341390442288732669957148001\\
V_\infty(4,68,2)=170605963060816377946936433265175680\\
V_\infty(4,72,2)=26829411396875692269491197638918648400\\
V_\infty(4,76,2)=4219165662049303123773116859323196816720\\
V_\infty(4,80,2)=663502408038018748448058464247159216890001\\
\\
V_\infty(5,5,2)=1\\
V_\infty(5,10,2)=252\\
V_\infty(5,15,2)=756756\\
V_\infty(5,20,2)=549676764\\
V_\infty(5,25,2)=298227062281\\
V_\infty(5,30,2)=218838390759073\\
V_\infty(5,35,2)=161446400503248672\\
V_\infty(5,40,2)=112632613848657302400\\
V_\infty(5,45,2)=79169699996993643966432\\
V_\infty(5,50,2)=56151546386557366024202177\\
V_\infty(5,55,2)=39717291593245217794329362081\\
V_\infty(5,60,2)=28058660061656964336359435570604\\
V_\infty(5,65,2)=19835819533825566529982592591911412\\
V_\infty(5,70,2)=14024417724324420598672399947721245804\\
V_\infty(5,75,2)=9914206081036463014882722168252570938889\\
V_\infty(5,80,2)=7008596284293402975749309111124669929079521\\
V_\infty(5,85,2)=4954676885097638926007640423100194180529855296\\
V_\infty(5,90,2)=3502659589845301193905028251874353899223998638208\\
V_\infty(5,95,2)=2476160267409321946445662150301548547825713614803904\\
V_\infty(5,100,2)=1750492069977099993617695861204414333904857504132837761$
\end{spacing}
}

\section*{Conclusion}
We extend Schwartz's result \cite{Sch09} to solve the ball size of frequency permutations under Chebyshev distance. We also offer another efficient algorithm for the cases of larger frequency, larger minimum distance and smaller symbol set.

\end{document}